\documentclass[final,1p,times,authoryear]{elsarticle}

\usepackage{booktabs}

\makeatletter
\def\ps@pprintTitle{%
    \let\@oddhead\@empty
    \let\@evenhead\@empty
    \def\@oddfoot{}%
    \let\@evenfoot\@oddfoot}
\makeatother

\usepackage{graphicx}
\usepackage{amssymb}
\usepackage{amsmath}
\usepackage{amsthm}
\usepackage{commath}
\usepackage{mathtools}
\usepackage{subcaption}
\usepackage{float}
\usepackage{lineno}
\usepackage{enumitem}
\usepackage{hyperref}

\graphicspath{ {./graphics/} }

\newtheorem{proposition}{Proposition}
\newtheorem{lemma}{Lemma}
\newtheorem{corollary}{Corollary}
\newtheorem{assumption}{Assumption}

\begin{document}

\begin{frontmatter}

\title{Robust Local Polynomial Regression with Similarity Kernels}
\author{Yaniv Shulman}
\ead{yaniv@shulman.info}

\begin{abstract}
Local Polynomial Regression (LPR) is a widely used nonparametric method for modeling complex relationships due to its flexibility and simplicity. It estimates a regression function by fitting low-degree polynomials to localized subsets of the data, weighted by proximity. However, traditional LPR is sensitive to outliers and high-leverage points, which can significantly affect estimation accuracy. This paper revisits the kernel function used to compute regression weights and proposes a novel framework that incorporates both predictor and response variables in the weighting mechanism. The focus of this work is a conditional density kernel that robustly estimates weights by mitigating the influence of outliers through localized density estimation. The proposed method is implemented in Python and is publicly available at \url{https://github.com/yaniv-shulman/rsklpr}. The population analysis quantifies the bias induced by density-based robust weighting, and the reported experiments show lower empirical bias than iterative robust LOWESS while remaining competitive with standard LOWESS. This advancement provides a promising extension to traditional LPR, opening new possibilities for robust regression applications.
\end{abstract}
\end{frontmatter}

\section{Introduction}
\label{s:Introduction}
Local polynomial regression (LPR) is a powerful and flexible statistical technique that has gained increasing popularity in recent years due to its ability to model complex relationships between variables. Local polynomial regression generalizes the polynomial regression and moving average methods by fitting a low-degree polynomial to a nearest neighbors subset of the data at the location. The polynomial is fitted using weighted ordinary least-squares, giving more weight to nearby points and less weight to points farther away. The value of the regression function for the point is then obtained by evaluating the fitted local polynomial using the predictor variable value for that data point. Local linear and local polynomial estimators have favorable boundary behavior and strong minimax-efficiency properties under standard smoothness conditions \citep{Fan1993,Fan1996Local}. The biggest advantage of this class of methods is not requiring a prior specification of a function i.e. a parameterized model. Instead, only a small number of hyperparameters need to be specified such as the type of kernel, a smoothing parameter and the degree of the local polynomial. The method is therefore suitable for modeling complex processes such as non-linear relationships, or complex dependencies for which no theoretical models exist. These two advantages, combined with the simplicity of the method, makes it one of the most attractive of the modern regression methods for applications that fit the general framework of least-squares regression but have a complex deterministic structure. 

Local polynomial regression incorporates the notion of proximity in two ways. The first is that a smooth function can be reasonably approximated in a local neighborhood by a simple function such as a linear or low order polynomial. The second is the assumption that nearby points carry more importance in the calculation of a simple local approximation or alternatively, that closer points are more likely to interact in simpler ways than far away points. This is achieved by a kernel which produces values that diminish as the distance between the explanatory variables increase to model stronger relationship between closer points. 

The relevant literature on classical and robust local regression is summarized in Section~\ref{S:RelatedWork}. LPR is susceptible to outliers, high leverage points and functions with discontinuities in their derivative, which can adversely affect least-squares based local estimates \citep{10.1002/wics.1492}.

The main contribution of this paper is to revisit the kernel used to produce regression weights. The simple yet effective idea is to generalize the kernel such that both the predictor and the response are used to calculate weights. Within this framework, a non-negative kernel based on conditional density estimation is proposed that assigns robust weights to mitigate the adverse effect of outliers in the local neighborhood. Note the proposed framework does not preclude the use of robust loss functions, robust bandwidth selectors and standardization techniques. In addition the method is implemented in the Python programming language and is made publicly available. In the experiments below, the proposed method substantially improves over iterative robust LOWESS in RMSE and bias diagnostics while remaining competitive with standard LOWESS, however robust LOWESS can still have lower MAE and median absolute error in the public benchmark.

The remainder of the paper is organized as follows: Section~\ref{S:RelatedWork} summarizes related work. In Section~\ref{S:Local Polynomial Regression}, a brief overview of the mathematical formulation of local polynomial regression is provided. In Section~\ref{S:Robust Weights with Similarity kernels}, a framework for robust weights and the specific conditional density kernel are proposed. Section~\ref{S:Properties} provides an analysis of the estimator and a discussion of its properties. In Section~\ref{S:Experiments and Implementation Notes}, implementation notes and experimental results are provided. Finally, in Section~\ref{S:Future work and research directions}, the paper concludes with directions for future research.

\section{Related Work}
\label{S:RelatedWork}

Local polynomial regression and LOWESS/LOESS estimate regression functions by fitting low-degree polynomials in local neighborhoods defined by predictor-space proximity. Classical kernel-regression and LOWESS/LOESS references establish the local-smoothing baseline, while local linear and local polynomial regression show why local polynomial fitting is often preferred over local averaging, especially for design adaptivity and reduced boundary bias \citep{Nadaraya1964OnER,Watson1964SmoothRA,cleveland79,cleveland_devlin88,Fan1993}.

Robust variants of local regression typically modify the fitting criterion or add residual-based reweighting so that large residuals have limited influence on the local fit. Cleveland's robust LOWESS is the classical example of iterative residual reweighting, while local \(M\)-estimation and adaptive robust local-polynomial methods provide related theoretical frameworks \citep{cleveland79,HardleGasser1984,ChichignoudLederer2014}; broader reviews are given in \citep{10.1002/wics.1492,SALIBIANBARRERA2023}.

The closest adjacent line to the proposed approach is the literature that uses response-space or conditional-density information in local fitting. Local modal regression introduces a response-space kernel and estimates conditional modes rather than conditional means, yielding robustness under skewed, heavy-tailed, or multimodal conditional distributions \citep{YaoLindsayLi2012}. In parallel, the conditional-density literature studies nonparametric estimators of \(f_{Y\mid X}(y\mid x)\) and associated bandwidth-selection methods \citep{HyndmanBashtannykGrunwald1996,BashtannykHyndman2001,HallRacineLi2004,CattaneoChandakJanssonMa2024}. RSKLPR sits between these lines: unlike residual-robust LOWESS, it uses estimated local response density rather than fitted residuals to form robustness weights; unlike local modal regression, it does not estimate the conditional mode, but retains a local-polynomial squared-error target whose oracle form is a density-weighted conditional mean. More robust density estimators, including divergence-based or RKHS-based robust KDE variants, are natural extensions left for future work \citep{BasuHarrisHjortJones1998,KimScott2012}.

\section{Local Polynomial Regression}
\label{S:Local Polynomial Regression}
This section provides a brief overview of local polynomial regression and establishes the notation subsequently used. We adopt the following standing assumptions: the training data $\mathcal{D}_T = \{(X_i,Y_i)\}_{i=1}^{T}$ are an i.i.d. sample from a continuous joint density $f_{X,Y}$; the error terms $\epsilon_i$ satisfy $\mathbb{E}[\epsilon_i | X_i] = 0$ and $\mathbb{E}[\epsilon_i^2 | X_i] = \sigma^2(X_i) < \infty$; the density of the predictors $f_X(x)$ is positive in the region of interest; and any kernel function $K$ is a non-negative, symmetric probability density function with finite second moments.

Let $( X, Y )$ be a random pair and $\mathcal{D}_T = \{(X_i,Y_i)\}_{i=1}^{T} \subseteq \mathcal{D}$ be a training set comprising a sample of $T$ data pairs. Suppose that $(X , Y) \sim f_{X,Y}$ a continuous density and $X \sim f_X$ the marginal distribution of $X$. Let $Y \in \mathbb{R}$ be a continuous response and assume a model of the form $Y_i=m(X_i) + \epsilon_i, \; i \in 1, \dots \, ,T$ where $m(\cdot): \mathbb{R}^d \rightarrow \mathbb{R }$ is an unknown function and $\epsilon_i$ are independently distributed error terms having zero mean such that $\mathbb{E}[Y \mid X=x] = m(x)$. There are no global assumptions about the function $m(\cdot)$ other than that it is smooth and that locally it can be well approximated by a low degree polynomial as per Taylor’s theorem. For notational simplicity, we present the one-dimensional case ($d=1$). The formulation extends to the multivariate case ($d>1$) by replacing powers with multi-indices (see, e.g., \citep{Fan1996Local}, \S3.2). The local $p$-th order Taylor expansion near a point $x$ yields:

\begin{align}
m(X_i) \approx \sum_{j=0}^p \frac{m^{(j)}(x)}{j!} (X_i - x)^j \coloneqq \sum_{j=0}^p \beta_j(x)(X_i - x)^j
\end{align}
Let \(r_p(t)=(1,t,\ldots,t^p)^\top\), so a local polynomial can be written as \(r_p(u-x)^\top\beta\).
To find an estimate $\hat{m}(x)$ of $m(x)$ the low-degree polynomial is fitted to the $N$ nearest neighbors using weighted least-squares such to minimize the empirical loss $\mathcal{L}_{\text{lpr}}(x; \mathcal{D}_N , h)$ :

\begin{align}
\mathcal{L}_{\text{lpr}}(x; \mathcal{D}_N , h) \coloneqq \sum_{i=1}^N \left(  Y_i - \sum_{j=0}^p \beta_j (x) (X_i-x)^j \right)^2 K_h(x-X_i)
\label{loss_lpr}
\end{align}
where $\beta(x) \in \mathbb{R}^{p+1}$ are the polynomial coefficients to be estimated. The minimizer is
\begin{align}
\hat{\beta}(x) \coloneqq \arg\min_{\beta(x)} \mathcal{L}_{\text{lpr}}(x; \mathcal{D}_N,h)
\label{beta_hat}
\end{align}
Where $K_h(\cdot) = h^{-d}K(\cdot/h)$ is a scaled kernel, $h \in \mathbb{R}_{> 0}$ is the bandwidth parameter and $\mathcal{D}_N \subseteq \mathcal{D}_T$ is the subset of $N$ nearest neighbors of $x$ in the training set where the distance is measured on the predictors only. Having computed $\hat{\beta}(x)$ the estimate of $m(x)$ is taken as $\hat{m}(x) = \hat{\beta}_0(x)$.

\begin{assumption}[WLS uniqueness]
\label{ass:wls_uniqueness}
Whenever a unique weighted least-squares coefficient vector is invoked, the active weighted local design matrix has full column rank. Equivalently, for the relevant nonnegative combined weights \(a_i\),
\[
\sum_i a_i r_p(X_i-x)r_p(X_i-x)^\top
\]
is nonsingular. In population objectives, the analogous weighted local-polynomial moment matrix is assumed nonsingular. For \(p=0\), this reduces to a positive total weight.
\end{assumption}

The term kernel carries here the meaning typically used in the context of nonparametric regression i.e. a non-negative real-valued weighting function that is typically symmetric, unimodal at zero, and integrates to one. Higher degree polynomials and smaller $N$ generally increase the variance and decrease the bias of the estimator and vice versa \citep{Fan1996Local,Garcia-Portugues2023}. For derivation of the local constant and local linear estimators for the multidimensional case see \citep{Garcia-Portugues2023}.

\paragraph{Remark on Nearest Neighbors and Bandwidth.}
In the following, the local neighborhood is defined by taking the \(N\) nearest neighbors to \(x\). Thus, \(\mathcal{D}_N \subseteq \mathcal{D}_T\) contains exactly \(N\) points. A distance-based kernel \(K_h\) is then used to weight those neighbors.

\section{Robust Weights with Similarity Kernels}
\label{S:Robust Weights with Similarity kernels}
The main idea presented is to generalize the kernel function used in equation \eqref{loss_lpr} to produce robust weights. This is achieved by using a similarity kernel function defined on the data domain $\mathcal{K}_{\mathcal{D}}:\mathcal{D}^2 \rightarrow \mathbb{R}_+$ that enables weighting each point and incorporating information on the data in the local neighborhood in relation to the local regression target $(x,y)$.
We use the term similarity kernel broadly: \(\mathcal K_{\mathcal D}\) may encode metric proximity between data points, but it may also encode other notions of compatibility or typicality.

The proposed empirical loss function is:
\begin{align}
\mathcal{L}_{\text{rsk}}(x, y; \mathcal{D}_N , \mathcal{H}) \coloneqq \sum_{i=1}^N \left(  Y_i - \sum_{j=0}^p \beta_j (x, y) (X_i - x)^j \right)^2 \mathcal{K}_{\mathcal{D}} \left((x,y),(X_i, Y_i); \mathcal{H} \right)
\label{loss_rsk}
\end{align}
The estimated coefficients are found by minimizing this loss:
\begin{align}
\hat{\beta}(x,y ; \mathcal{D}_N, \mathcal{H} ) \coloneqq \arg\min_{\beta(x,y)} \mathcal{L}_{\text{rsk}}(x, y; \mathcal{D}_N , \mathcal{H}) 
\label{beta_loss_d}
\end{align}
Where $\mathcal{H}$ is the set of bandwidth parameters. There are many possible choices for such a similarity kernel to be defined within this general framework. However, used as a local weighting function, such a kernel should have the following attributes:
\begin{enumerate}
\item Non-negative, $\mathcal{K}_{\mathcal{D}}((x,y) , ( x', y')) \geq 0$.
\item Symmetry in the inputs, $\mathcal{K}_{\mathcal{D}}((x,y) , ( x', y')) = \mathcal{K}_{\mathcal{D}}((x', y') , (x,y))$.
\item Tending toward decreasing as the distance in the predictors increases. That is, given  a similarity function on the response $s(\cdot, \cdot): \mathbb{R}^2 \rightarrow \mathbb{R}_+$, if $s(y, y')$ indicates high similarity the weight should decrease as the distance between the predictors grows, $s(y, y') > \alpha \implies      
\mathcal{K}_{\mathcal{D}}((x, y) , ( x + a, y')) \geq \mathcal{K}_{\mathcal{D}}((x, y) , (x + b, y')) \; \; \forall \: \norm{a} \leq \: \norm{b}$ and some $\alpha \in \mathbb{R}_+$.
\end{enumerate}

In this work a useful non-negative kernel is proposed. Similarly to the usual kernels used in \eqref{loss_lpr}, these tend to diminish as the distance between the explanatory variables increases to model stronger relationship between closer points. In addition, the weights produced by the kernels also model the "importance" of the pair $(x,y)$. This is useful for example to down-weight outliers to mitigate their adverse effect on the ordinary least square based regression. Note that for the Reproducing Kernel Hilbert Space (RKHS) interpretation discussed in Section \ref{S:Properties}, the kernel $\mathcal{K}_{\mathcal{D}}$ must also be positive-definite, but this condition is not required for the main results of this paper. Formally let $\mathcal{K}_{\mathcal{D}}$ be defined as:
\begin{align}
\mathcal{K}_{\mathcal{D}} \left((x,y),(x', y') ; \mathcal{H}_1, \mathcal{H}_2 \right) &= K_1(x,x' ; \mathcal{H}_1) K_2 \left((x,y),(x', y') ; \mathcal{H}_2 \right)
\label{compound_kernel_def}
\end{align}
Where $K_1: \mathbb{R}^d \times \mathbb{R}^d \rightarrow \mathbb{R}_+$ and $K_2: \mathcal{D}^2 \rightarrow \mathbb{R}_+$ are non-negative kernels and $\mathcal{H}_1$, $\mathcal{H}_2$ are the sets of bandwidth parameters. The purpose of $K_1$ is to account for the distance between a neighbor to the local regression target and therefore may be chosen as any of the kernel functions that are typically used in equation \eqref{loss_lpr}. The role of $K_2$ is to perform robust regression by detecting local outliers in an unsupervised manner and assigning them with lower weights.

The material below gives a kernel-agnostic lemma that shows when the optimisation for the empirical estimator $\hat\beta(x)$ is invariant to the (unknown) response value $y$ at the regression location $x$. A corollary then specialises this result to the conditional-density kernel, which is the focus of this paper.

\begin{lemma}[Invariance under separable similarity kernels]
\label{lemma:invariance}
Let the similarity kernel be
\[
\mathcal{K}_{\mathcal{D}}\!\bigl((x,y),(x',y');\mathcal{H}\bigr)
= K_1\!\bigl(x,x';\mathcal{H}_1\bigr)\,
K_2\!\bigl((x,y),(x',y');\mathcal{H}_2\bigr),
\]
with $K_1$ being any non-negative kernel function on $\mathbb{R}^d \times \mathbb{R}^d$, and let $K_2$ be separable:
\[
K_2\!\bigl((x,y),(x',y');\mathcal{H}_2\bigr)
= c(x,y)\,w(x',y'),
\qquad \text{where } c(x,y)>0 \text{ and } w(x',y')\ge 0.
\]
Assume Assumption~\ref{ass:wls_uniqueness} holds with weights
\[
a_i=K_1(x,X_i;\mathcal{H}_1)w(X_i,Y_i).
\]
Then the empirical loss \eqref{loss_rsk} becomes
\[
\mathcal{L}_{\text{rsk}}(x,y;\mathcal{D}_N,\mathcal{H})
= c(x,y)\!\sum_{i=1}^{N}
\Bigl(Y_i-\textstyle\sum_{j=0}^{p}\beta_j(x)(X_i - x)^j\Bigr)^2
K_1(x,X_i;\mathcal{H}_1)\,w(X_i,Y_i),
\]
so the weighted least-squares minimiser with respect to $\beta_j$ is unique, independent of $y$, and will be denoted $\hat\beta(x) = (\hat\beta_0(x), \dots, \hat\beta_p(x))^T$. For a general separable kernel \(K_2(z,z')=c(z)w(z')\), symmetry requires
\[
c(z)w(z')=c(z')w(z)
\qquad \text{for all } z,z'.
\]
In the nontrivial case, this means that \(c\) is proportional to \(w\) on the relevant support. The choice \(c=w\) is the canonical special case used by the conditional-density product kernel. Symmetry is not required for the weighted least-squares minimization problem itself.
\end{lemma}
The proof is given in~\ref{appendix:proof_invariance}.

\subsection*{Conditional Density Kernel}
The primary method proposed for $K_2$ is proportional to the conditional density of the response.
The conditional-density product kernel used below is a separable example of this broader class; after cancellation of the target-dependent factor, it acts as a local density-based reweighting of the training observations.
The ideal conditional-density kernel is defined using the oracle conditional density of the response:
\[
K_2\!\bigl((x,y),(x',y')\bigr)
=
f_{Y\mid X}(y\mid x)\,
f_{Y\mid X}(y'\mid x').
\]

In finite samples, the unknown conditional density is replaced by a plug-in estimate \(\hat f_{Y\mid X}\), yielding
\[
K_2\!\bigl((x,y),(x',y');\mathcal{H}_2\bigr)
=\hat f_{Y\mid X}(y\mid x; \mathcal{H}_2)\,
\hat f_{Y\mid X}(y'\mid x'; \mathcal{H}_2),
\]
where $\hat f_{Y\mid X}(\cdot\mid\cdot; \mathcal{H}_2)$ is an estimated conditional density of the response with bandwidth(s) $\mathcal{H}_2$.
Given a density-estimation sample \(\mathcal D_M^{\mathrm{kde}}=\{(X_\ell,Y_\ell)\}_{\ell=1}^{M}\), with positive definite bandwidth matrices \(\mathbf H_{xy}\) and \(\mathbf H_x\) and positive denominator below, the plug-in conditional-density estimate is
\begin{equation}
\hat f_{Y\mid X}(y\mid x;\mathcal H_2)
=
\frac{
|\mathbf H_{xy}|^{-1/2}
\sum_{\ell=1}^{M}
K_{xy}\!\left(
\mathbf H_{xy}^{-1/2}
\left(
\begin{bmatrix}x\\y\end{bmatrix}
-
\begin{bmatrix}X_\ell\\Y_\ell\end{bmatrix}
\right)
\right)
}{
|\mathbf H_x|^{-1/2}
\sum_{\ell=1}^{M}
K_x\!\left(
\mathbf H_x^{-1/2}(x-X_\ell)
\right)
}.
\label{conditional_prob_est}
\end{equation}
In the finite-sample nearest-neighbor implementation, \(\mathcal D_M^{\mathrm{kde}}=\mathcal D_N(x)\), so the plug-in density estimate is target-specific; the empirical density weight for \((X_i,Y_i)\) is more precisely \(\hat f^{(x)}_{Y\mid X}(Y_i\mid X_i)\). The cancellation argument below is read conditionally on the selected neighborhood for the current target \(x\): the resulting weights are independent of the unknown response value \(y\), but need not define a single global symmetric kernel on \(\mathcal D^2\) unless the density-estimation sample is fixed independently of the target point.

\begin{corollary}[Conditional–density kernel objective]
\label{cor:cond_density_obj}
Choose $K_1(x,x';\mathcal{H}_1)$ to be a standard kernel for local polynomial regression, such as $K_{h_1}(x-x')$, and let $K_2$ be the finite-sample plug-in conditional density kernel defined above.
Then, we can identify
\[
c(x,y)=\hat f_{Y\mid X}(y\mid x; \mathcal{H}_2),
\qquad \text{and} \qquad
w(X_i,Y_i)=\hat f_{Y\mid X}(Y_i\mid X_i; \mathcal{H}_2).
\]
Assuming $\hat f_{Y\mid X}(y\mid x; \mathcal{H}_2) > 0$, at least one combined training weight is positive, and Assumption~\ref{ass:wls_uniqueness} holds for the active weighted local design, Lemma~\ref{lemma:invariance} yields the simplified weighted least-squares objective for $\hat\beta(x)$:
\[
\tilde{\mathcal{L}}(x)
=\sum_{i=1}^{N}
\Bigl(Y_i-\sum_{j=0}^{p}\beta_j(x)(X_i - x)^j\Bigr)^2
K_{h_1}(x-X_i)\,\hat f_{Y\mid X}(Y_i\mid X_i; \mathcal{H}_2),
\]
which is the empirical objective function used by the finite-sample implementation.
\end{corollary}

\begin{figure}[!htbp]
\centering
\begin{subfigure}[t]{0.36\textwidth}
\centering
\includegraphics[width=\textwidth]{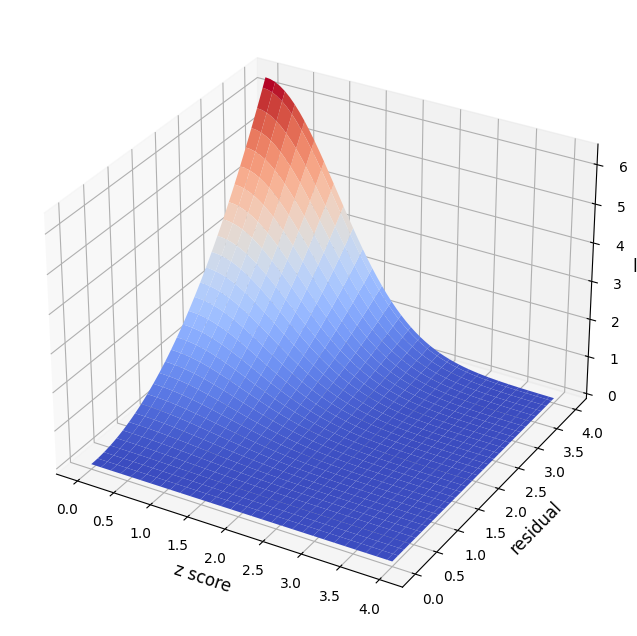}
\caption{Loss surface over residual and response value.}
\end{subfigure}
\hfill
\begin{subfigure}[t]{0.48\textwidth}
\centering
\includegraphics[width=\textwidth]{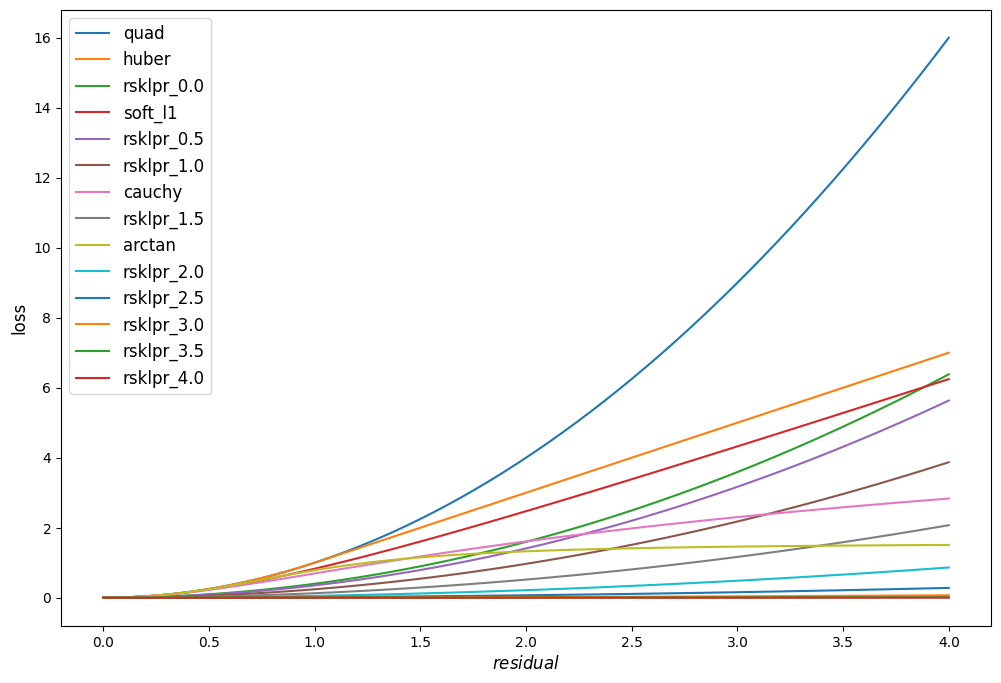}
\caption{Slices compared with quadratic, Huber, and Cauchy losses.}
\end{subfigure}
\caption{Diagnostic visualizations of the density-weighted squared-error loss. Panel (a) shows residual on the horizontal axis and response value on the depth axis; panel (b) shows one-dimensional slices at the indicated numbers of response standard deviations from the mean. Both plots use a quadratic residual loss, a standard Gaussian response-density proxy for $K_2$, and omit the $K_1$ distance-kernel scaling. The vertical axis is proportional to loss $\times$ density, illustrating attenuation in low-density response regions.}
\end{figure}

Regardless of the choice of kernel, the hyperparameters of this model are similar in essence to the standard local polynomial regression and comprise the span of included points, the kernels and their associated bandwidths. Note that the density estimator can be replaced with other robust density estimators however exploring this option is left for future work.

\section{Properties}
\label{S:Properties}

This section discusses the properties of the proposed estimator, beginning with its interpretation on a finite sample and then moving to an idealized population-target analysis. The population analysis below is separate from the finite-sample nearest-neighbor implementation used in the experiments.

\subsection{Finite-Sample Interpretation as a Re-weighted LPR}
At the sample level, the proposed estimator can be understood as a direct re-weighting of the terms in the standard LPR loss function. The weights are determined by the local conditional density of the response.

\begin{proposition}[Equivalence to a Re-weighted LPR Objective]
\label{prop:reweighting}
Minimizing the proposed empirical loss from Corollary \ref{cor:cond_density_obj},
\[
\tilde{\mathcal{L}}(x)
=\sum_{i=1}^{N}
\Bigl(Y_i-\sum_{j=0}^{p}\beta_j(x)(X_i - x)^j\Bigr)^2
K_{h_1}(x-X_i)\,\hat f_{Y\mid X}(Y_i\mid X_i),
\]
is equivalent to minimizing a weighted average of the standard LPR loss terms, where each term's contribution is scaled by its estimated conditional density $\hat f_{Y\mid X}(Y_i\mid X_i)$. Assume Assumption~\ref{ass:wls_uniqueness} holds with weights
\[
a_i=K_{h_1}(x-X_i)\hat f_{Y\mid X}(Y_i\mid X_i).
\]
\end{proposition}
The proof is given in~\ref{appendix:proof_reweighting}.

\subsection{Asymptotic Properties}
We now analyze the idealized population objective corresponding to the density-weighted loss. Throughout this subsection, \(w(u,v)\) denotes the population-level non-negative weight function, with the conditional-density target obtained from the oracle weight
\[
w(u,v)=f_{Y\mid X}(v\mid u).
\]

\begin{assumption}[Local weighted centering for the local-linear case]
\label{ass:local_weighted_centering}
For the local-linear case \(p=1\), assume that the first centered weighted design moment vanishes:
\[
M_1(x)
=
\iint
(u-x)K_{h_1}(x-u)w(u,v)f_{X,Y}(u,v)\,\mathrm d v\,\mathrm d u
=
0.
\]
When \(u,x\in\mathbb R^d\), this is understood componentwise. Equivalently, the weighted design density
\[
A(u)=\int w(u,v)f_{X,Y}(u,v)\,\mathrm d v
\]
is locally centered by the kernel around \(x\). A sufficient condition is that both \(K_{h_1}\) and \(A(u)\) are symmetric about \(x\); kernel symmetry alone, or conditional response-density symmetry alone, is insufficient.
\end{assumption}

This condition is used only to reduce the \(p=1\) local-linear intercept to the simple ratio form below; no such assumption is required for \(p=0\). It may be plausible at interior points when \(A(u)\) varies slowly over the kernel window, but can fail near boundaries or under strongly asymmetric designs.

\begin{proposition}[Population Objective and the Intercept Term]
\label{prop:population_objective}
Let $f_{X,Y}(u,v)$ denote the joint density of $(X,Y)$ where $X \in \mathbb{R}^d, Y \in \mathbb{R}$. For a chosen regression point $x \in \mathbb{R}^d$, define the population objective function as
\[
\mathcal J(x;\beta)
=\iint_{\mathbb R^{d}\!\times\!\mathbb R}
\bigl(v-g_x(u;\beta)\bigr)^{2}
K_{h_1}(x-u)\,
w(u,v)\,
f_{X,Y}(u,v)\,\mathrm d v\,\mathrm d u,
\]
where $g_x(u;\beta)$ is the local polynomial, $K_{h_1}$ is a kernel function, and $w(u,v)$ is a population-level non-negative weight function. Assume the weighted second moment is finite:
\[
\iint v^2K_{h_1}(x-u)w(u,v)f_{X,Y}(u,v)\,\mathrm d v\,\mathrm d u<\infty.
\]
Using the basis \(r_p\) defined in Section~\ref{S:Local Polynomial Regression}, define the weighted local-polynomial design moment matrix
\[
S_p(x)=
\iint r_p(u-x)r_p(u-x)^\top
K_{h_1}(x-u)w(u,v)f_{X,Y}(u,v)\,\mathrm d v\,\mathrm d u.
\]
For \(p\ge1\), assume \(S_p(x)\) is nonsingular.
Let $\beta^{\star}(x)=\arg\!\min_{\beta}\,\mathcal J(x;\beta)$. The first component, $\beta^{\star}_0(x)$, represents the local intercept of the polynomial fit at $x$. Define
\begin{align}
M_0(x)
&=\iint
K_{h_1}(x-u)\,
w(u,v)\,
f_{X,Y}(u,v)\,\mathrm d v\,\mathrm d u, \label{eq:M0}\\
R_0(x)
&=\iint v\,
K_{h_1}(x-u)\,
w(u,v)\,
f_{X,Y}(u,v)\,\mathrm d v\,\mathrm d u. \label{eq:R0}
\end{align}
Assume $M_0(x)>0$.
For local constant regression ($p=0$), with $g_x(u;\beta)=\beta_0$, the population minimizer satisfies
\begin{equation}
\beta^{\star}_0(x)=\frac{R_0(x)}{M_0(x)}. \label{eq:pop_intercept}
\end{equation}
For the local-linear case ($p=1$), write
\[
g_x(u;\beta)=\beta_0+\beta_1^\top(u-x),
\]
and suppose Assumption~\ref{ass:local_weighted_centering} holds. Then the population minimizer also satisfies \eqref{eq:pop_intercept}. The ratio characterization is restricted to the local constant case and to the local-linear case under Assumption~\ref{ass:local_weighted_centering}.
\end{proposition}
The proof is given in~\ref{appendix:proof_population_objective}. For higher-order local polynomials ($p \ge 2$), the simple ratio formula is generally not valid. The unique intercept should instead be obtained as the first component of the full weighted normal-equation solution based on \(S_p(x)\). We therefore restrict the explicit ratio characterization to the local constant case and to the locally centered local-linear case.

For the oracle conditional-density target, define
\[
C(u)=\int [f_{Y\mid X}(v\mid u)]^2\,\mathrm d v,
\qquad
\mu'(u)=
\frac{\int v[f_{Y\mid X}(v\mid u)]^2\,\mathrm d v}{C(u)}.
\]

\begin{assumption}[Local squared-density moments]
\label{ass:local_squared_density_moments}
For \(u\) in a neighborhood of \(x\),
\[
\begin{gathered}
0<C(u)<\infty,\\
\int |v|[f_{Y\mid X}(v\mid u)]^2\,\mathrm d v<\infty,\\
\int v^2[f_{Y\mid X}(v\mid u)]^2\,\mathrm d v<\infty.
\end{gathered}
\]
\end{assumption}

The finite-sample implementation replaces this oracle weight by a kernel conditional-density estimate \(\hat f_{Y\mid X}(v\mid u)\), as described in Section~\ref{S:Experiments and Implementation Notes}.

\begin{proposition}[Asymptotic Conditional-Density Target]
\label{prop:conditional_density_asymptotic_target}
Suppose Assumption~\ref{ass:local_squared_density_moments} holds. The oracle conditional-density population criterion is equivalent, up to a \(\beta\)-independent term, to ordinary local polynomial regression of \(\mu'(u)\) with design weight \(C(u)f_X(u)\). Under standard local-polynomial assumptions, including an interior point \(x\), shrinking bandwidth, continuity and positivity of \(C(u)f_X(u)\), and sufficient smoothness of \(\mu'\) near \(x\), the intercept of the population local-polynomial fit converges to \(\mu'(x)\).
\end{proposition}
The proof is given in~\ref{appendix:proof_conditional_density_asymptotic_target}.

\begin{corollary}[Population Target with Conditional Density Weights]
\label{cor:population_target}
Consider the specific case where the weight function is the true conditional density, $w(u,v) = f_{Y\mid X}(v\mid u)$, and suppose \(C(u)>0\) and the integrability conditions of Assumption~\ref{ass:local_squared_density_moments} hold on the relevant neighborhood of \(x\).
Assume further that \(C(u)f_X(u)\) and \(\mu'(u)C(u)f_X(u)\) are continuous at \(x\), and that \(x\) is an interior point or otherwise satisfies the relevant boundary regularity conditions. In the cases covered by Proposition \ref{prop:population_objective}, the population intercept $\beta_0^\star(x)$ from \eqref{eq:pop_intercept} takes the explicit form:
\[
\beta_0^\star(x) = \frac{\iint v\,K_{h_1}(x-u)\,[f_{Y\mid X}(v\mid u)]^2\,f_X(u)\,\mathrm d v\,\mathrm d u}{\iint K_{h_1}(x-u)\,[f_{Y\mid X}(v\mid u)]^2\,f_X(u)\,\mathrm d v\,\mathrm d u}.
\]
This expression can be rewritten as a locally weighted average of $\mu'(u)$:
\[
\beta_0^\star(x) = \frac{\int K_{h_1}(x-u)\,\mu'(u)\,C(u)\,f_X(u)\,\mathrm d u}{\int K_{h_1}(x-u)\,C(u)\,f_X(u)\,\mathrm d u}.
\]
This is the exact finite-bandwidth ratio from Proposition~\ref{prop:population_objective} for the local constant case and for the locally centered local-linear case. Proposition~\ref{prop:conditional_density_asymptotic_target} gives the corresponding asymptotic local-polynomial target, \(\mu'(x)\), for the full local-polynomial fit. This target is generally different from the true conditional mean $m(x) = \mathbb{E}[Y|X=x]$ and provides the formal basis for the asymptotic bias discussion.
\end{corollary}

\subsection{Asymptotic Target and Conditions for Unbiasedness}
Proposition~\ref{prop:conditional_density_asymptotic_target} establishes that the oracle population objective has pointwise target $\mu'(x)$. A crucial question is under what conditions this target coincides with the true regression function, $m(x) = \mathbb{E}[Y|X=x]$. The exact condition is
\[
\int (v-m(x))[f_{Y\mid X}(v\mid x)]^2\,\mathrm d v=0,
\]
or equivalently $\mu'(x)=m(x)$.

A simple sufficient condition is conditional symmetry around $m(x)$. This includes the normal distribution, but also any symmetric conditional density (e.g., Laplace, Student's t). If for each fixed $x$, the conditional density $f_{Y\mid X}(v\mid x)$ is symmetric around $m(x)$, then $[f_{Y\mid X}(v\mid x)]^2$ is also symmetric around $m(x)$, so the squared-density expectation remains $m(x)$ and therefore $\mu'(x)=m(x)$. Under conditional symmetry and standard localization conditions, the pointwise asymptotic target of the proposed population objective therefore coincides with the conditional mean \(m(x)\). However, the finite-bandwidth objective is not identical to the standard LPR objective: when \(w(u,v)=f_{Y\mid X}(v\mid u)\), integrating over \(v\) introduces the effective factor \(C(u)\), so neighboring covariate values are weighted by \(f_X(u)C(u)\), rather than by \(f_X(u)\) alone. Thus, even under conditional symmetry, the two objectives should be understood as having the same local pointwise target in the limit \(h_1\to0\), not as being equivalent at finite bandwidth. For example, under Gaussian conditional errors with variance \(\sigma^2(u)\), \(C(u)\) is proportional to \(1/\sigma(u)\), so heteroscedasticity changes the effective local weighting.

When this squared-density unbiasedness condition fails, the mean under the squared density $\mu'(X)$ differs from the true mean $m(X)$, introducing an asymptotic bias of $\text{Bias}(x) = \mu'(x) - m(x)$. An example quantifying this bias for the asymmetric exponential distribution is provided in \ref{appendix:asym_bias_example}.

\subsection{Comparison with Standard and Iterative Robust LPR}
While the proposed robust method builds on the LPR framework, its weighting mechanism introduces key differences.

\subsubsection{The Core Difference vs. Standard LPR: The Weighting Function}
The fundamental difference lies in what determines the "importance" of a neighboring data point $(u,v)$ when estimating the regression function at a point $x$. For standard LPR, the population objective aims to minimize:
    \[\mathcal{J}_{\text{std}}(x;\beta) = \iint \bigl(v-g_x(u;\beta)\bigr)^{2} K_{h_1}(x-u) f_{X,Y}(u,v)\,\mathrm d v\,\mathrm d u\]
    The weight is determined by the kernel $K_{h_1}(x-u)$ and the data-generating process, but it is linear in the conditional density term $f_{Y|X}(v|u)$. \newline \newline
For the proposed method (with $w(u,v) = f_{Y|X}(v|u)$), the population objective is:
    \[\mathcal{J}_{\text{rsk}}(x;\beta) = \iint \bigl(v-g_x(u;\beta)\bigr)^{2} K_{h_1}(x-u) [f_{Y|X}(v|u)]^2 f_X(u)\,\mathrm d v\,\mathrm d u\]
    The proposed method's key innovation is the squaring of the conditional density term, $[f_{Y|X}(v|u)]^2$. This change amplifies the weighting effect, more strongly down-weighting observations $(u,v)$ where the response $v$ is unlikely given the predictor $u$.

\subsubsection{The True Counterpart: Iterative Robust Methods}
The direct practical counterpart to the proposed method is iterative robust LPR, such as the procedure used in LOWESS. These methods use an \emph{iterative approach} by repeatedly fitting the data and adjusting weights. After each fit, residuals are calculated, and new "robustness weights" are assigned to each point, typically by down-weighting points with large residuals. In contrast, the proposed method is a single-step procedure where the weights are derived from an explicit estimate of the data-generating distribution itself. 

It is understood that many robust estimators can introduce some bias as a price for their resilience to outliers. While iterative robust LPR is also subject to such biases, this aspect often receives insufficient attention in the literature, largely due to the analytical challenges involved. In contrast, the proposed method, by virtue of its non-iterative nature and direct link to the data distribution, makes this trade-off explicit. The bias towards $\mu'(x)$ is clearly defined and can be analyzed, offering a degree of theoretical transparency that is not readily available for its iterative counterparts.

\subsection{Trade-off Between Robustness and Bias via the \texorpdfstring{\( K_2 \)}{K2} Kernel and Bandwidth Selection}
The proposed estimator utilizes the \( K_2 \) kernel to adjust data point weights based on both predictors and responses, controlling the trade-off between robustness and bias. The bandwidth \( \mathcal{H}_2 \) of the \( K_2 \) kernel plays a crucial role in this mechanism.

In the finite-sample plug-in objective, the density-only weight is \(w_i=\hat{f}(Y_i \mid X_i; \mathcal{H}_2)\), giving the combined weight \(a_i=K_{h_1}(x-X_i)w_i\). The \( K_2 \) component assigns lower weights to responses with smaller estimated conditional density.

In finite samples, extremely small density bandwidths may reduce the contrast among plug-in density weights, while very large \( \mathcal{H}_2 \) makes the density estimator nearly constant across different \( Y_i \) and the estimator approaches standard LPR. An intermediate bandwidth \( \mathcal{H}_2 \) achieves a balance. Implementation-specific details are given in Section~\ref{S:Experiments and Implementation Notes}.

\subsection{Relationship to Kernel Methods and RKHS}
The RKHS interpretation is ancillary to the estimator but provides useful intuition for the density-based weights. Suppose the KDE smoothing kernel \(k\) is positive definite, with feature map \(\phi\) into an RKHS \(\mathcal H\), so that
\[
k(z,z')=\langle \phi(z),\phi(z')\rangle_{\mathcal H}.
\]
For a density-estimation sample \(\{Z_\ell\}_{\ell=1}^{M}\), the KDE at \(z\) can be written as
\[
\hat f(z)
=
\frac{1}{M}\sum_{\ell=1}^{M} k(z,Z_\ell)
=
\left\langle
\phi(z),
\frac{1}{M}\sum_{\ell=1}^{M}\phi(Z_\ell)
\right\rangle_{\mathcal H}.
\]
Thus, the KDE is the inner product between the feature representation of the evaluation point and the empirical mean feature embedding of the local data \citep{Muandet2017}. For the conditional-density weight
\[
\hat f_{Y\mid X}(y\mid x)=\hat f_{X,Y}(x,y)/\hat f_X(x),
\]
the numerator has this mean-embedding interpretation in the joint \((X,Y)\)-space, while the denominator normalizes by the predictor-space density.

From this perspective, the KDE acts as a point-to-set similarity score: it compares the feature representation of the evaluation point \(z\) with the empirical mean feature embedding of the local density-estimation sample.

The conditional-density weight can therefore be interpreted as a normalized point-to-set similarity in the joint \((X,Y)\)-space, normalized by the corresponding point-to-set similarity in predictor space.

This perspective suggests that the KDE score could in principle be replaced by other point-to-local-cloud similarity scores, such as robust feature-mean embeddings or other distributional similarity measures. We focus here on conditional-density and joint-density weights because they provide a direct probabilistic interpretation and yield an explicit oracle population objective.

This KDE mean-embedding interpretation is distinct from, but compatible with, the separable product kernel used for weighting. The product-kernel statement applies to the oracle score or to a plug-in score computed from a fixed density-estimation sample. If \(g(z)=\hat f_{Y\mid X}(y\mid x)\), then the conditional-density product kernel can be written as
\[
K_2(z,z')=g(z)g(z'),
\]
which is a symmetric rank-one positive semidefinite kernel. If \(K_1\) is also positive semidefinite, then \(K_1K_2\) is positive semidefinite by standard kernel closure properties. For the weighted least-squares estimator itself, however, positive semidefiniteness is not required; nonnegative weights are sufficient.

\section{Experiments and Implementation Notes}
\label{S:Experiments and Implementation Notes}
The proposed method (RSKLPR) is implemented in Python and published as an open source package \href{https://github.com/yaniv-shulman/rsklpr}{https://github.com/yaniv-shulman/rsklpr}. This section reports three complementary empirical studies. First, a public regression benchmark evaluates predictive performance under cross-validated hyperparameter selection. Second, a synthetic target-validation experiment tests whether the empirical behaviour agrees with the population-target analysis in Section~\ref{S:Properties}. Third, a controlled target-corruption experiment on the public Appliances dataset examines how the methods respond to zero-mean symmetric and asymmetric label corruption.

The experiments in this section are intended to evaluate the behavior of the proposed density-based local weighting mechanism relative to standard LOWESS and robust LOWESS under controlled and public benchmark settings. They are not intended to establish state-of-the-art performance for the underlying prediction tasks, nor to compare against the broad range of specialized models that may be better suited to a particular dataset. The focus is instead on whether RSKLPR preserves the local-regression structure while avoiding the degradation observed with residual-based robust LOWESS in the settings considered.

\subsection{Implementation Notes}
The implementation normalizes distances in each neighborhood to the range $[0, 1]$, consistent with the approach in \cite{cleveland79}, effectively making the bandwidth for $K_1$ adaptive. The population analysis above is formulated for an idealized kernel-weighted objective; the implementation uses a fixed-size nearest-neighbor neighborhood, together with distance normalization and kernel weighting, as a practical adaptive localization scheme. The predictor-distance kernel is specified separately for each experiment. For density estimation in $K_2$, a factorized multidimensional Kernel Density Estimator (KDE) with scaled Gaussian kernels was used \citep{Rosenblatt1956,Parzen1962}. As noted above, the finite-sample KDE is target-specific in the nearest-neighbor implementation. This self-inclusive plug-in KDE is evaluated on the same local neighborhood points used to estimate it; the resulting density weights are mean-normalized and clipped away from zero for numerical stability. With extremely small density bandwidths, each observation's self-kernel contribution can dominate its estimated density, which may reduce the contrast among density weights. Leave-one-out or regularized density weights are natural alternatives. Full finite-sample and asymptotic analysis of the nearest-neighbor localization and plug-in weighting is left for future work. The implementation supports multiple KDE bandwidth-selection rules, and the bandwidth rule is therefore treated as an experiment-level specification; the rules used in the reported experiments are given in the corresponding appendix sections. This factorization is a simplification that ignores potential covariance between predictors and response but is computationally efficient.

The experiments reported below are a curated subset of a broader exploratory study. Additional synthetic experiments, including homoscedastic and heteroscedastic Gaussian noise, asymmetric response distributions, sparse and dense sampling regimes, multivariate regression surfaces, neighborhood-sensitivity checks, and bootstrap-based uncertainty estimates, are available as notebooks in the \href{https://nbviewer.org/github/yaniv-shulman/rsklpr/tree/main/src/experiments}{project experiments directory}. These notebooks were used to guide the reported benchmark design, while the following sections focus on the experiments most directly tied to the estimator's empirical performance and population-target analysis.

\subsection{Public Regression Benchmark}
We evaluated the methods on the UCI Appliances Energy Prediction dataset \citep{UCIAppliancesEnergyPrediction,CandanedoFeldheimDeramaix2017}. This dataset contains sensor and weather measurements from a residential building, with the appliance energy use as the response.

The compared methods were standard LOWESS, robust LOWESS, and RSKLPR with the conditional-density kernel. All methods used local linear fits and tricube predictor-distance weights. Hyperparameters were selected by inner cross-validation within each training split. The final focused grid was chosen after an initial coarse sweep: standard LOWESS and RSKLPR used the Minkowski metric with $p=1$, while robust LOWESS used the Mahalanobis metric after a separate metric check. The neighborhood size was selected from $\{35,41,\ldots,119\}$ for each method and split. The full preprocessing, splitting, and hyperparameter protocol is given in~\ref{appendix:public_benchmark_protocol}.

\begin{table}[t]
\caption{Focused public benchmark results on the UCI Appliances Energy Prediction dataset. Values summarize five repeated 80/20 holdout splits after inner four-fold cross-validation. Lower RMSE, MAE, and median absolute error are better; Win is the fraction of held-out splits with the lowest RMSE.}
\centering
\scriptsize
\begin{tabular}{lrrrrrr}
\toprule
Method & Mean RMSE & Med. RMSE & RMSE SD & Mean MAE & Med. AE & Win \\
\midrule
RSKLPR & 74.878 & 75.298 & 1.553 & 36.968 & 12.945 & 0.60 \\
LOWESS & 75.137 & 74.943 & 1.793 & 37.530 & 13.398 & 0.40 \\
Robust LOWESS & 84.376 & 84.594 & 1.285 & 35.981 & 10.411 & 0.00 \\
\bottomrule
\end{tabular}
\label{tab:appliances_benchmark}
\end{table}

The results are shown in Table~\ref{tab:appliances_benchmark} and Figure~\ref{fig:appliances_benchmark}. RSKLPR achieved the lowest mean RMSE and won three of the five held-out splits, although the mean-RMSE margin over standard LOWESS was modest. The main comparison is with robust LOWESS, which applies residual-based robust reweighting. In this benchmark, robust LOWESS substantially underperforms both standard LOWESS and RSKLPR in RMSE, whereas RSKLPR preserves a robustness-motivated reweighting structure without incurring the same loss in predictive accuracy. The difference between RSKLPR and standard LOWESS is modest, so we interpret the result primarily as evidence that the proposed density-based robust weighting avoids the degradation observed for residual-based robust LOWESS.

\begin{figure}[t]
\caption{Focused public benchmark on the UCI Appliances Energy Prediction dataset. Left: mean held-out RMSE with one standard deviation over five repeated holdout splits. Right: split-wise held-out RMSE distribution. Hyperparameters were selected by inner four-fold cross-validation within each split.}
\centering
\begin{subfigure}[t]{0.48\textwidth}
\centering
\includegraphics[width=\textwidth]{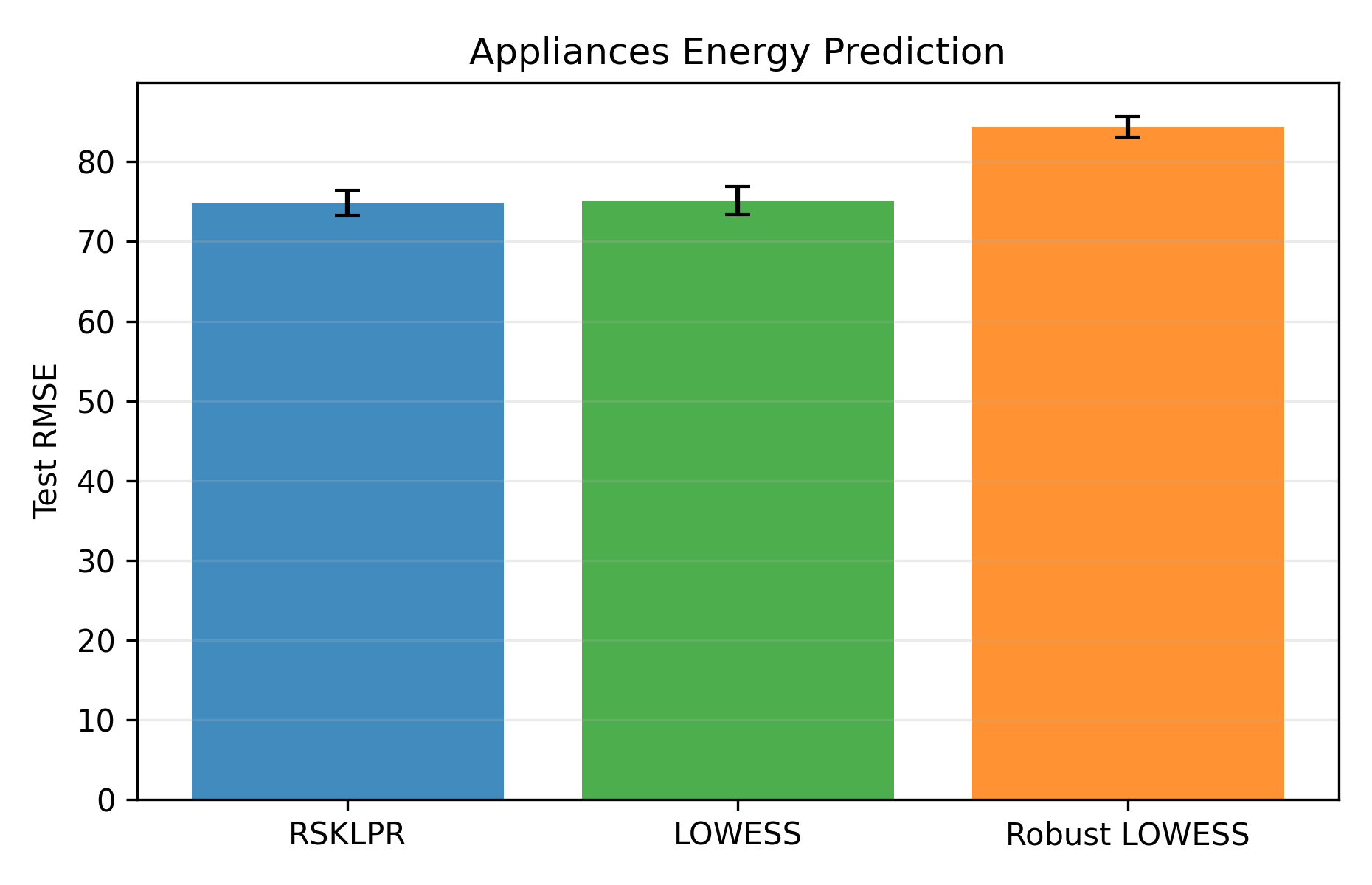}
\end{subfigure}
\hfill
\begin{subfigure}[t]{0.48\textwidth}
\centering
\includegraphics[width=\textwidth]{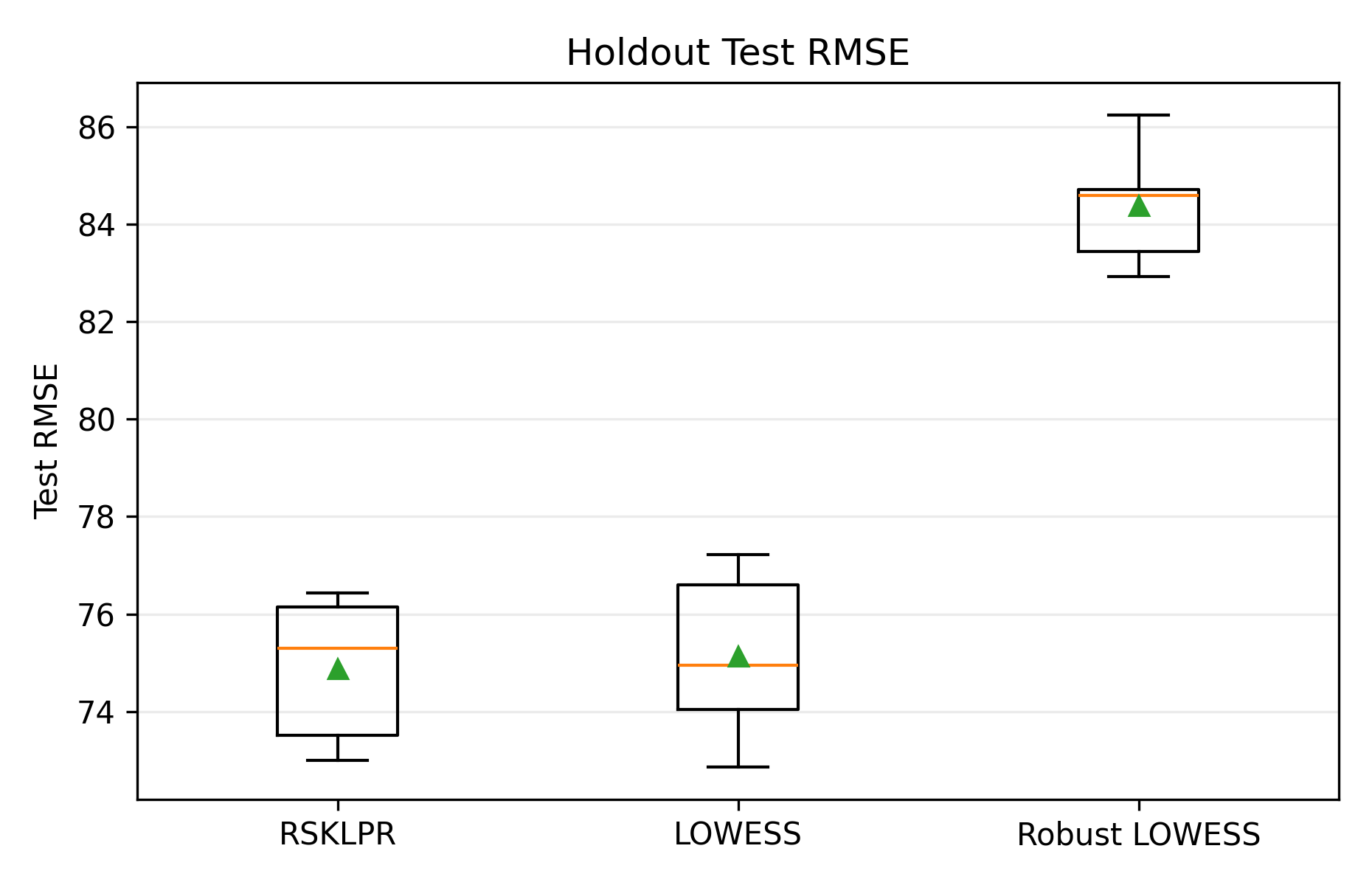}
\end{subfigure}
\label{fig:appliances_benchmark}
\end{figure}

\subsection{Synthetic Target-Validation Experiment}
The second experiment validates the theoretical target discussed in Section~\ref{S:Properties}. The purpose is not to claim broad benchmark superiority, but to test whether the empirical behaviour agrees with the population-target analysis: under symmetric conditional response distributions RSKLPR should target the conditional mean, while under asymmetric conditional response distributions it should exhibit the density-tilted bias described in Corollary~\ref{cor:population_target}.

Synthetic datasets were generated from a smooth one-dimensional regression function $m(x)$ with known ground truth; the precise data-generating mechanism is given in~\ref{appendix:experimental_hyperparameters}. Six conditional response distributions were considered: two symmetric distributions (Gaussian and a symmetric bimodal Gaussian mixture) and four asymmetric positive distributions (Exponential, Gamma, Log-normal, and Weibull). The symmetric bimodal distribution is included to test a non-Gaussian case where the conditional distribution remains symmetric about $m(x)$ but is not unimodal. Representative low- and high-density samples from selected distributions are shown in Figure~\ref{fig:density_example_samples}.

For each distribution and each training-set size \(T \in \{50, 75, 100, 150, 250, 500, 1000, 2000\}\), 50 Monte Carlo trials were generated. Each method was fitted on the sampled noisy training data and evaluated on a fixed clean grid of 300 points using RMSE against the true conditional mean $m(x)$. The plotted curves report the mean RMSE across trials, with shaded 95\% Monte Carlo intervals. The $x$-axis is shown on a logarithmic scale to make the low-density regime visible.

The methods compared are standard LOWESS, robust LOWESS (five robustness iterations), RSKLPR using the conditional-density kernel, and the joint-density RSKLPR variant described in~\ref{appendix:joint_density_kernel}. LOWESS is included primarily as a theoretical reference for the conditional mean under clean zero-mean sampling. Iterative robust LOWESS is the main robustness comparator.

All methods used the same local neighborhood size at each training-set size: approximately one quarter of the training sample, with a minimum of 15 neighbors and a maximum of 200. The LOWESS baselines used the corresponding local span, so each method was fitted with the same effective local sample size; robust LOWESS used five residual reweighting iterations. Additional implementation details are provided in~\ref{appendix:experimental_hyperparameters}.

Figure~\ref{fig:density_target_validation} shows the results. The symmetric cases behave as predicted. Under Gaussian noise, all methods converge toward the true conditional mean as the number of points increases. The symmetric bimodal case is more challenging at low data density, but remains a clean test of the symmetry condition: the response distribution is multimodal, yet centered on $m(x)$.

For the asymmetric distributions, standard LOWESS generally has the lowest RMSE against $m(x)$, which is expected because it directly targets the conditional mean under clean sampling. This curve should therefore be interpreted as a conditional-mean reference rather than as the primary robustness comparison. The main comparison is with iterative robust LOWESS. At smaller training sizes (\(T\le150\)), the RSKLPR variants have lower RMSE and narrower Monte Carlo intervals than robust LOWESS in all panels except the bimodal case. This suggests that iterative residual reweighting can introduce unnecessary bias when the asymmetry in the response distribution is genuine structure rather than contamination.

At larger sample sizes in the asymmetric settings, the RSKLPR curves flatten near the density-tilted reference rather than converging to the LOWESS conditional-mean target. This empirically supports the population-target analysis in Corollary~\ref{cor:population_target} and the exponential example in~\ref{appendix:asym_bias_example}. RSKLPR uses a single non-iterative density-reweighting step, making the associated bias more explicit and theoretically analyzable than the iterative residual reweighting in robust LOWESS.

\begin{figure}[t]
\caption{Empirical target validation under clean synthetic sampling. Each panel shows mean RMSE against the true conditional mean as a function of training-set size \(T\) over 50 Monte Carlo trials; shaded bands show 95\% Monte Carlo intervals. The dashed horizontal line shows the theoretical asymptotic RMSE implied by the oracle density-tilted population target and is a reference level rather than a bound. LOWESS is a conditional-mean reference baseline. In the symmetric Gaussian and bimodal cases, the reference level is zero and all methods trend toward the conditional-mean target. At smaller training sizes (\(T\le150\)), the RSKLPR variants have lower RMSE and narrower Monte Carlo intervals than iterative robust LOWESS in all panels except the bimodal case. In the asymmetric positive-response cases, LOWESS continues to move toward the conditional mean, while the RSKLPR conditional- and joint-density curves flatten at nonzero error levels near the oracle density-tilted reference, up to finite-sample variability.}
\centering
\includegraphics[width=1.0\textwidth]{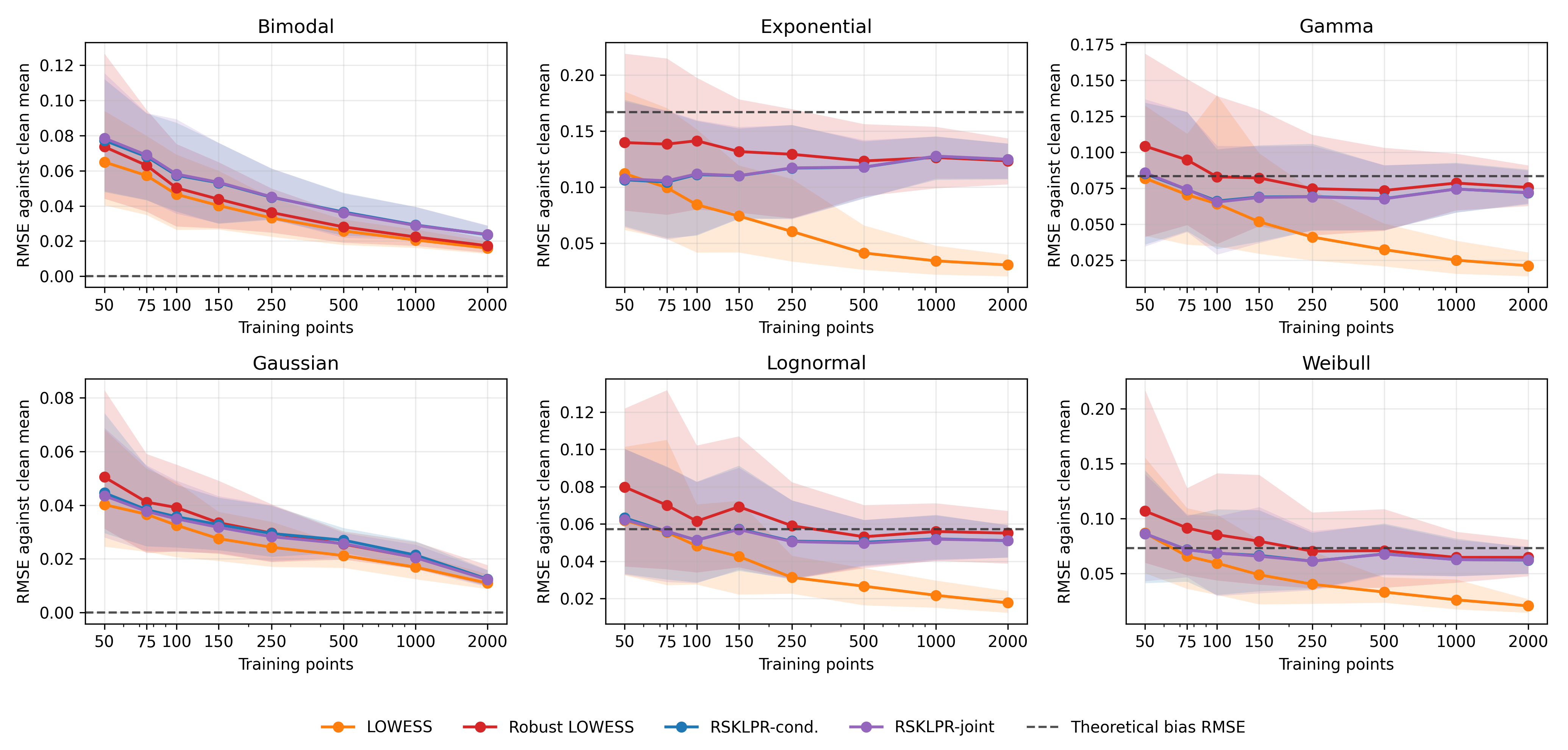}
\label{fig:density_target_validation}
\end{figure}

\subsection{Controlled Target-Corruption Experiment}
The public Appliances benchmark suggested that robust LOWESS and RSKLPR respond differently to high-response observations. To isolate this behaviour from the unknown noise process in the original dataset, we ran a controlled target-corruption experiment on the same UCI Appliances data. The clean response was retained as a reference target. For each outer training split, zero-mean corruption was added only to the training responses, while model selection and held-out evaluation were scored against the clean response. This protocol is therefore a recovery diagnostic rather than a deployment setting with unobserved clean validation labels. Full experimental details are given in~\ref{appendix:corrupted_appliances_protocol}.

Five corruption profiles were considered: Gaussian, Student-\(t(3)\), centered exponential, centered log-normal, and centered sparse one-sided contamination. Each sampled corruption vector was centered and standardized before being scaled by a fraction of the clean training-response standard deviation. Thus the asymmetric profiles are asymmetric in shape but not in mean. This distinction is important because standard LOWESS remains a useful conditional-mean diagnostic under zero-mean target noise.

Figure~\ref{fig:corrupted_appliances_diagnostic} summarizes the nonzero corruption levels from the coarse grid run. LOWESS remained stable under both symmetric and centered asymmetric corruption, indicating that the injected target outliers were largely averaged out by the local neighborhoods and did not create a severe leverage problem. The more informative comparison is therefore with iterative robust LOWESS. Bias is computed as the mean clean-target residual, \(n_{\mathrm{test}}^{-1}\sum_i(\hat y_i-y_i^{\mathrm{clean}})\); for each corruption profile, level, and method, this signed bias is averaged over splits, and the reported absolute bias values average its absolute value over the relevant nonzero corruption settings. Under centered asymmetric corruption, robust LOWESS had substantially larger absolute bias than RSKLPR (19.689 versus 0.262 on average) and substantially worse clean-target RMSE (85.154 versus 78.436). Under symmetric and heavy-tailed corruption the same pattern remained, though the robust LOWESS bias was less severe (11.318 versus 0.576).

These results suggest that, in this setting, iterative residual reweighting overreacts to asymmetric residual tails and down-weights genuine high-energy appliance observations, improving typical absolute errors but shifting the fitted function downward. RSKLPR does not show the same asymmetric residual-feedback bias: it preserves near-zero bias and clean-target RMSE while still providing a response based robust counterpart to LOWESS. This does not contradict the population-target analysis above, which concerns genuinely asymmetric conditional response distributions. Instead, it shows that for centered exogenous target corruption, the density-weighted estimator can behave much closer to the conditional-mean target than the worst-case asymptotic bias discussion might suggest.

\begin{figure}[t]
\caption{Controlled target corruption on the UCI Appliances Energy Prediction dataset. Results average over nonzero corruption levels and five repeated holdout splits. Symmetric/heavy-tail includes Gaussian and Student-\(t(3)\) corruption; centered asymmetric includes centered exponential, centered log-normal, and centered sparse one-sided contamination. The absolute-bias axis is shown on a logarithmic scale.}
\centering
\includegraphics[width=0.86\textwidth]{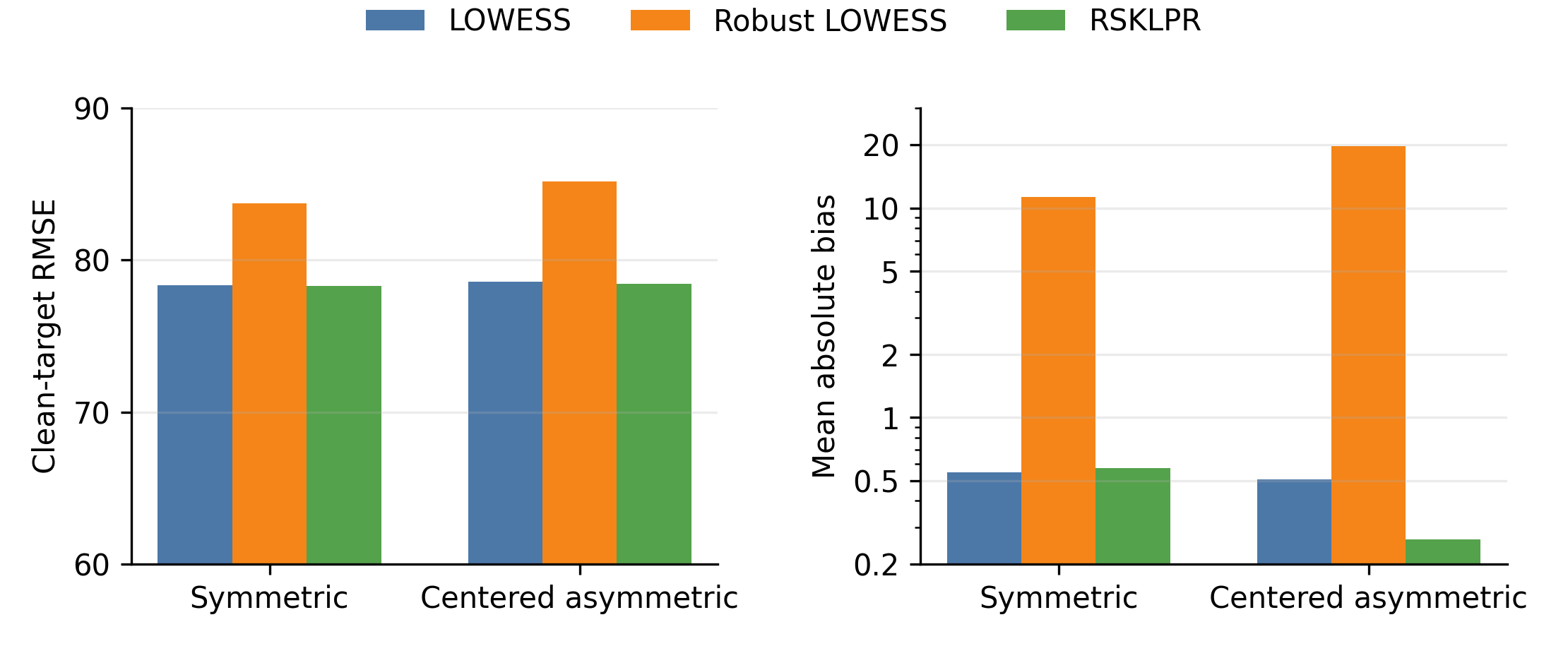}
\label{fig:corrupted_appliances_diagnostic}
\end{figure}

\section{Future Work and Research Directions}
\label{S:Future work and research directions}

This work introduces a new robust variant of Local Polynomial Regression (LPR), opening several avenues for further exploration and refinement. Since the proposed method generalizes the traditional LPR, there are opportunities to replace certain standard components in equation \eqref{loss_rsk} with more robust alternatives. These could include approaches such as robust methods for bandwidth selection or substituting the conventional quadratic residual function with alternatives better suited for handling outliers.

An important research direction is to explore adaptive bandwidth selection strategies that respond dynamically to local data density. In regions where data are sparse, the bandwidth in \(K_2\) could be fine-tuned to maintain robust down-weighting of potential outliers. Conversely, in denser regions, broader bandwidths may be adopted, causing the estimator to behave more like standard LPR and reduce any bias introduced by the robust weighting. Incorporating such adaptive bandwidths could further enhance the method’s overall performance and flexibility.

Additionally, further development of this framework may involve exploring different kernel functions and assessing how robust density estimators influence overall performance. Extending the method within the RKHS framework presents another valuable direction. This could allow for the introduction of a regularization term in the loss function, enhancing control over estimator smoothness and mitigating the risk of overfitting. Through these future directions, the robustness and adaptability of the proposed method could be substantially advanced.

\clearpage

\bibliographystyle{elsarticle-harv}
\bibliography{refs}

\clearpage
\appendix
\section{Proofs}
\label{appendix:proofs}

\subsection{Proof of Lemma~\ref{lemma:invariance}}
\label{appendix:proof_invariance}
\begin{proof}
The term \( c(x,y) \), which is positive and constant with respect to the summation index \(i\), is a scalar factor multiplying the entire sum. Assumption~\ref{ass:wls_uniqueness} makes the weighted least-squares objective strictly convex in the local polynomial coefficients, so the argmin is a singleton. Since scaling an objective by a positive constant preserves this singleton argmin, the minimizer \(\hat\beta(x,y)\) is independent of \(y\) and can be denoted by \(\hat\beta(x)\).
\end{proof}

\subsection{Proof of Proposition~\ref{prop:reweighting}}
\label{appendix:proof_reweighting}
\begin{proof}
Let $w_i = \hat f_{Y\mid X}(Y_i\mid X_i)$ and \(a_i=K_{h_1}(x-X_i)w_i\). Under Assumption~\ref{ass:wls_uniqueness}, the weighted least-squares objective is strictly convex in the local polynomial coefficients, so the minimizer is unique. Since \(\sum_{i=1}^N w_i>0\), dividing the objective by this positive sum does not change the minimizer:
\begin{align}
\hat{\beta}(x) &= \arg\min_{\beta(x)} \frac{1}{\sum_{k=1}^N w_k} \sum_{i=1}^N w_i \left[ \left(Y_i-\sum_{j=0}^{p}\beta_j(x)(X_i - x)^j\right)^2 K_{h_1}(x-X_i) \right].
\end{align}
The term in the square brackets is the $i$-th term of the standard LPR loss function from Equation \eqref{loss_lpr}. The expression is therefore a weighted arithmetic mean of these standard LPR terms. This interpretation makes it clear that points with a low estimated conditional density (i.e., response outliers) are down-weighted in a single, non-iterative step. Note that if a kernel with bounded support (e.g., Epanechnikov) is used for density estimation, it is theoretically possible for all weights $w_i$ in a neighborhood to be zero, although this is not an issue with unbounded kernels like the Gaussian.
\end{proof}

\subsection{Proof of Proposition~\ref{prop:population_objective}}
\label{appendix:proof_population_objective}
\begin{proof}
The weighted moments used in the normal-equation calculation are
\begin{align}
M_0(x)
&=\iint
K_{h_1}(x-u)\,
w(u,v)\,
f_{X,Y}(u,v)\,\mathrm d v\,\mathrm d u, \label{eq:proof_M0}\\
M_1(x)
&=\iint
(u-x)K_{h_1}(x-u)\,
w(u,v)\,
f_{X,Y}(u,v)\,\mathrm d v\,\mathrm d u, \label{eq:M1_centering}\\
R_0(x)
&=\iint v\,
K_{h_1}(x-u)\,
w(u,v)\,
f_{X,Y}(u,v)\,\mathrm d v\,\mathrm d u. \label{eq:proof_R0}
\end{align}

For $p=0$, the objective is
\[
\mathcal J(x;\beta_0)
=\iint
(v-\beta_0)^2
K_{h_1}(x-u)\,
w(u,v)\,
f_{X,Y}(u,v)\,\mathrm d v\,\mathrm d u.
\]
Differentiating with respect to $\beta_0$ and setting the derivative to zero gives
\[
0
=-2\iint
(v-\beta_0^\star)
K_{h_1}(x-u)\,
w(u,v)\,
f_{X,Y}(u,v)\,\mathrm d v\,\mathrm d u
=-2\{R_0(x)-\beta_0^\star(x)M_0(x)\}.
\]
Thus $\beta_0^\star(x)=R_0(x)/M_0(x)$ whenever $M_0(x)>0$.

For $p=1$, write $g_x(u;\beta)=\beta_0+\beta_1^\top(u-x)$. Differentiating with respect to the intercept gives the first normal equation
\begin{equation}
\beta_0^\star(x)M_0(x)+(\beta_1^\star(x))^\top M_1(x)=R_0(x). \label{eq:first_normal_equation}
\end{equation}
When $M_1(x)=0$, Equation \eqref{eq:first_normal_equation} reduces to $\beta_0^\star(x)=R_0(x)/M_0(x)$. When $M_1(x)\neq0$, the local-linear intercept instead contains the correction term
\[
\beta_0^\star(x)
=\frac{R_0(x)-(\beta_1^\star(x))^\top M_1(x)}{M_0(x)}
=\frac{R_0(x)}{M_0(x)}
-\frac{(\beta_1^\star(x))^\top M_1(x)}{M_0(x)}.
\]
\end{proof}

\subsection{Proof of Proposition~\ref{prop:conditional_density_asymptotic_target}}
\label{appendix:proof_conditional_density_asymptotic_target}
\begin{proof}
With the basis \(r_p\) defined in Section~\ref{S:Local Polynomial Regression}, write \(g_x(u;\beta)=r_p(u-x)^\top\beta\). With the oracle conditional-density
weight, \(f_{X,Y}(u,v)=f_{Y\mid X}(v\mid u)f_X(u)\) gives
\[
\mathcal J(x;\beta)
=\int K_{h_1}(x-u)f_X(u)
\left\{\int
\bigl(v-r_p(u-x)^\top\beta\bigr)^2
[f_{Y\mid X}(v\mid u)]^2\,\mathrm d v\right\}\mathrm d u.
\]
For fixed \(u\), the inner integral expands as
\[
\int
\bigl(v-r_p(u-x)^\top\beta\bigr)^2
[f_{Y\mid X}(v\mid u)]^2\,\mathrm d v
=
C(u)\bigl(\mu'(u)-r_p(u-x)^\top\beta\bigr)^2
+B(u),
\]
where
\[
B(u)=
\int
\bigl(v-\mu'(u)\bigr)^2
[f_{Y\mid X}(v\mid u)]^2\,\mathrm d v
\]
does not depend on \(\beta\). Therefore minimizing \(\mathcal J(x;\beta)\)
is equivalent, up to the \(\beta\)-independent term
\(\int K_{h_1}(x-u)f_X(u)B(u)\,\mathrm d u\), to minimizing
\[
\int
K_{h_1}(x-u)
\bigl(\mu'(u)-r_p(u-x)^\top\beta\bigr)^2
C(u)f_X(u)\,\mathrm d u.
\]
This is the ordinary population local-polynomial criterion for the regression
function \(\mu'(u)\) with design weight \(C(u)f_X(u)\). Under the stated
local-polynomial assumptions, the weighted kernel moment matrix converges after
rescaling to \(C(x)f_X(x)\) times the limiting kernel moment matrix, and the
smoothness of \(\mu'\) near \(x\) yields the standard local-polynomial
population expansion. Hence the intercept of the population fit converges to
\(\mu'(x)\) as \(h_1\to0\).
\end{proof}

\section{Experimental Protocols}

\subsection{Public Benchmark Protocol}
\label{appendix:public_benchmark_protocol}
The public benchmark in Section~\ref{S:Experiments and Implementation Notes} used the UCI Appliances Energy Prediction dataset, available from the \href{https://archive.ics.uci.edu/dataset/374/appliances+energy+prediction}{UCI Machine Learning Repository}. The response was appliance energy use. The timestamp was converted to cyclic hour-of-day and day-of-week features, after which the original timestamp was dropped. The two random variables included in the dataset (\texttt{rv1} and \texttt{rv2}) were also dropped. Within each outer training split, the eight predictors with largest absolute Pearson correlation with the response were selected using training data only; this avoids test-set leakage. The selected features were stable across splits and consisted of the hour features, lights, outdoor humidity, selected indoor/outdoor temperatures, and indoor humidity.

The evaluation used five repeated holdout splits with an 80/20 train/test split. Each split contained 15,788 training observations and 3,947 test observations. Hyperparameters were selected by four-fold cross-validation on the training portion of each split, using RMSE as the selection criterion. The selected estimator was then refit on the full outer training split before held-out test evaluation. All numerical predictors were standardized using training-split statistics before fitting.

The benchmark proceeded in two stages. A coarse grid first searched over predictor kernels, distance metrics, response-kernel choices, and a broad set of neighborhood sizes. This selected the tricube predictor-distance kernel, Minkowski distance with $p=1$ for standard LOWESS and RSKLPR, and the conditional-density kernel (\texttt{kr=conden}) for RSKLPR. Robust LOWESS was then checked over the same focused neighborhood grid with Minkowski $p=1$, Minkowski $p=2$, and Mahalanobis distance; Mahalanobis was selected consistently in the observed splits. The final focused benchmark therefore used tricube predictor-distance weights for all methods; Minkowski $p=1$ for LOWESS and RSKLPR; Mahalanobis distance for robust LOWESS; no residual reweighting for LOWESS; three bisquare residual-reweighting iterations for robust LOWESS; and the conditional-density kernel for RSKLPR. For all three methods, the final neighborhood grid was $35,41,\ldots,119$ with step size 6.

The LOWESS baseline is a multivariate local linear regression baseline with tricube distance weights and no residual or response-density reweighting. Robust LOWESS uses the same local linear model and tricube distance weighting, followed by iterative bisquare residual reweighting. The RSKLPR model instead performs a single non-iterative response-density reweighting step using the conditional-density kernel, with normal-reference bandwidths for the response-density KDE \citep{Silverman1986}.

\subsection{Synthetic Target-Validation Protocol}
\label{appendix:experimental_hyperparameters}
The synthetic target-validation experiment in Section~\ref{S:Experiments and Implementation Notes} was run with fixed hyperparameters chosen to make the comparison interpretable rather than to tune each method separately. For training size \(T\), predictors were generated on the domain \(x\in[0,1]\) from an equally spaced grid with independent Gaussian jitter \(N(0,(0.15/T)^2)\), then clipped to \([0,1]\) and sorted. The clean regression function was
\[
m(x)
=
\sqrt{\left|x^3-\frac{4}{3}x^4\right|}
+0.1x\sin^2(3\pi x)
-\min_{t\in[0,1]}\left\{
\sqrt{\left|t^3-\frac{4}{3}t^4\right|}
+0.1t\sin^2(3\pi t)
\right\}
+0.1,
\]
so \(m(x)>0\) on the simulation domain. This positivity is required for the positive asymmetric response distributions, which were parameterized to have conditional mean \(m(x)\): Exponential responses used scale \(m(x)\); Gamma responses used shape \(2\) and scale \(m(x)/2\); Log-normal responses used \(\sigma=0.5\) and \(\mu(x)=\log m(x)-\sigma^2/2\); and Weibull responses used shape \(k=1.5\) and scale \(m(x)/\Gamma(1+1/k)\). The Gaussian response had mean \(m(x)\) with heteroscedastic scale proportional to the response range, and the bimodal response was a symmetric two-component perturbation around \(m(x)\).

All methods used local linear fits ($p=1$) and the same neighborhood size $N_{\mathrm{loc}}$ for a given training-set size. The resulting neighborhood sizes were
\[
\begin{array}{c|cccccccc}
T & 50 & 75 & 100 & 150 & 250 & 500 & 1000 & 2000 \\
\hline
N_{\mathrm{loc}} & 15 & 19 & 25 & 38 & 62 & 125 & 200 & 200
\end{array}
\]
where the final two entries reflect the maximum-neighborhood cap.

The LOWESS and robust LOWESS baselines were computed in one dimension with the same effective local sample size as RSKLPR. Standard LOWESS used the tricube distance kernel with no residual reweighting; robust LOWESS used the same local span with five bisquare residual-reweighting iterations.

The RSKLPR variants used the same $N_{\mathrm{loc}}$ and local linear degree. The predictor-distance kernel $K_1$ was the normalized Laplacian kernel applied after scaling neighborhood distances to $[0,1]$. The conditional-density and joint-density response kernels used the normal-reference bandwidth rule in the local KDE. No cross-validation or distribution-specific hyperparameter tuning was used in this synthetic validation experiment.

The dashed reference line in Figure~\ref{fig:density_target_validation} is the asymptotic RMSE against the true conditional mean implied by Corollary~\ref{cor:population_target}. For each response distribution, the squared-density target can be written as $\rho m(x)$, where $m(x)$ is the clean conditional mean. The plotted value is therefore $|1-\rho|\{\int m(x)^2 dx\}^{1/2}$, evaluated on the same clean grid as the experiment. The ratios are $\rho=1$ for the Gaussian and symmetric bimodal cases, $\rho=1/2$ for the Exponential case, $\rho=3/4$ for the Gamma case with shape 2, $\rho=\exp(-3\sigma^2/4)$ for the Log-normal case with $\sigma=0.5$, and $\rho=2^{-1/k}/\{\Gamma(2-1/k)\Gamma(1+1/k)\}$ for the Weibull case with shape $k=1.5$.

\begin{figure}[p]
\caption{Representative synthetic target-validation samples for selected response distributions. The black curve is the true conditional mean $m(x)$ and the points are sampled training responses. The left column shows the low-density regime (\(T=50\)), while the right column shows the high-density regime (\(T=2000\)). The Exponential and Log-normal cases illustrate the asymmetric positive response distributions used to test the density-tilted population target.}
\centering
\includegraphics[width=1.0\textwidth]{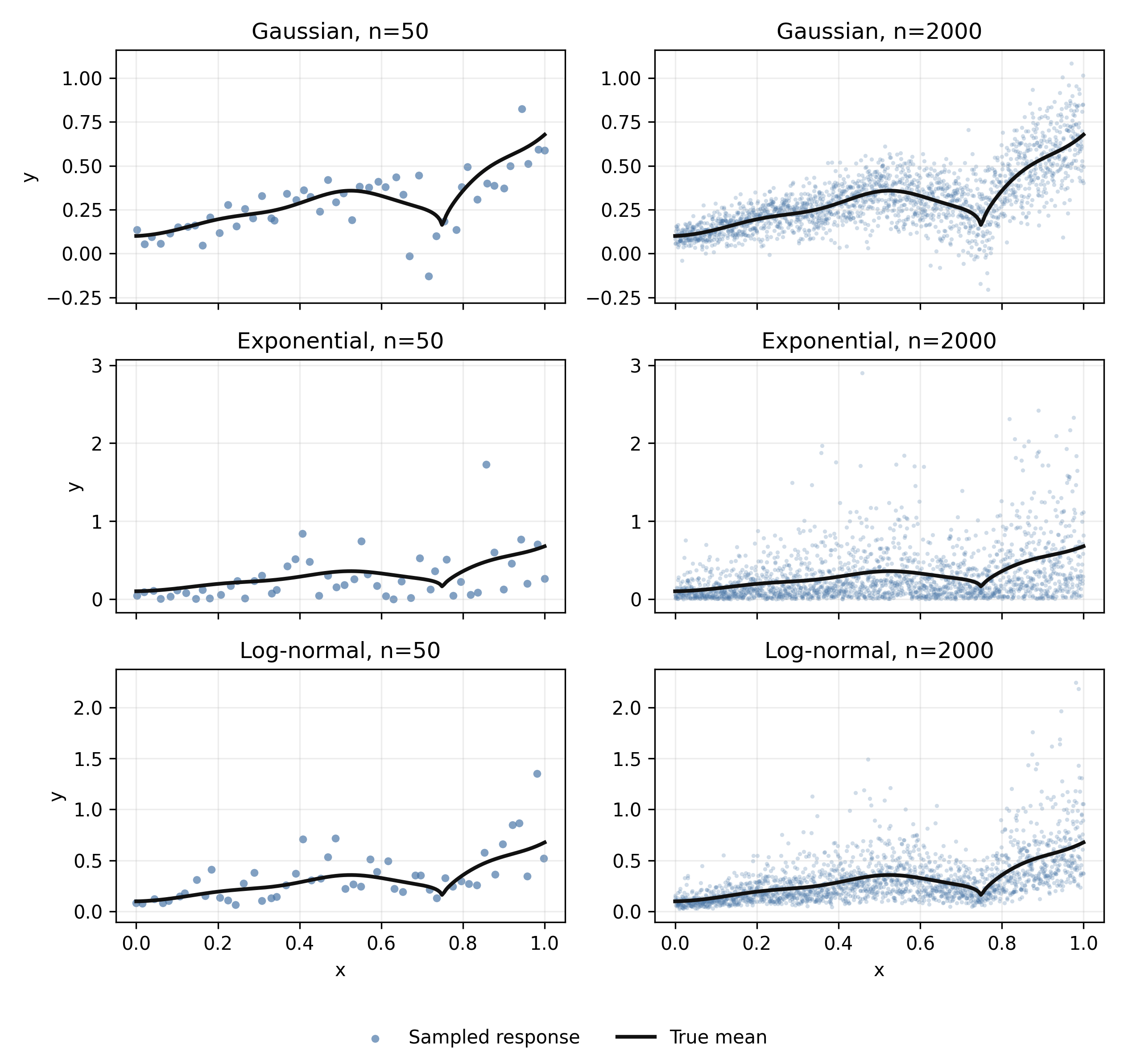}
\label{fig:density_example_samples}
\end{figure}

\subsection{Controlled Target-Corruption Protocol}
\label{appendix:corrupted_appliances_protocol}
The controlled target-corruption experiment used the same UCI Appliances preprocessing as Appendix~\ref{appendix:public_benchmark_protocol}: cyclic time features were added, \texttt{rv1} and \texttt{rv2} were dropped, the eight predictors with largest absolute training-split Pearson correlation with the clean response were selected in each outer split, and predictors were standardized using training-split statistics. The evaluation again used five repeated 80/20 holdout splits with four-fold inner cross-validation.

The corruption was applied only to the outer training response. Let \(y_i\) denote the clean training response and let \(s_y\) be its training-split standard deviation. For each noise profile, an i.i.d. noise vector \(z_i\) was sampled, centered by subtracting its empirical mean, and standardized by its empirical standard deviation. The corrupted training response was then
\[
 y_i^{\mathrm{corr}} = y_i + \alpha s_y z_i,
\]
where \(\alpha \in \{0,0.1,0.25,0.5,1.0\}\). The noise profiles were Gaussian, Student-\(t(3)\), exponential, log-normal, and sparse one-sided contamination. The contamination profile used small Gaussian background noise and added exponential positive shocks to 5\% of observations before centering and standardization. Consequently, the exponential, log-normal, and contamination settings are centered asymmetric corruption profiles rather than positive-mean corruption processes.

Models were fitted on \(y^{\mathrm{corr}}\), but both inner-CV model selection and outer held-out evaluation were scored against the original clean response. This clean-target scoring is intentionally an oracle-style recovery benchmark: it asks how each method behaves when the training labels are corrupted by a known zero-mean process while the desired target remains the clean conditional mean.

The reported run used a coarse hyperparameter grid. All methods used local linear fits and tricube predictor-distance weights. Standard LOWESS used \(\texttt{kr=none}\); RSKLPR searched over the conditional-density and joint-density kernels, \(\texttt{kr} \in \{\texttt{conden},\texttt{joint}\}\); and robust LOWESS used three bisquare residual-reweighting iterations. For standard LOWESS and RSKLPR, the distance metric was selected from Minkowski \(p=1\), Minkowski \(p=2\), and Mahalanobis distance. Robust LOWESS used the same metric choices. The neighborhood size was selected from \(\{31,63,95,127\}\). For each corruption type, level, split, and method, the selected estimator was refit on the full corrupted outer training response and evaluated against the clean held-out response.

\section{Joint Density Kernel}
\label{appendix:joint_density_kernel}
An alternative kernel for $K_2$ can be defined that is proportional to the joint distribution of the random pair. This could be useful, for example, to also down-weight high-leverage points in the predictor space.
\begin{align}
K_2 \left((x,y),(x', y') ; \mathcal{H}_2 \right) = \hat{f}(x,y ; \mathcal{H}_2) \hat{f}(x',y'; \mathcal{H}_2)
\end{align}
where the joint density can be estimated using the Parzen-Rosenblatt window estimator \citep{Rosenblatt1956,Parzen1962}. This choice also satisfies the conditions of Lemma \ref{lemma:invariance}, with $c(x,y)=\hat{f}(x,y; \mathcal{H}_2)$ and $w(X_i, Y_i) = \hat{f}(X_i, Y_i; \mathcal{H}_2)$. The simplified empirical objective function for a point $(X_i, Y_i)$ in the neighborhood becomes:
\[
\tilde{\mathcal{L}}(x)
=\sum_{i=1}^{N}
\Bigl(Y_i-\sum_{j=0}^{p}\beta_j(x)(X_i - x)^j\Bigr)^2
K_{h_1}(x-X_i)\,\hat{f}(X_i, Y_i; \mathcal{H}_2).
\]
This formulation weights each point $(X_i, Y_i)$ by its estimated joint density, in addition to the standard distance-based weight $K_{h_1}(x-X_i)$.

The mechanism by which this kernel provides robustness becomes clearer when we consider its population-level objective. The objective function involves an integral term weighted by $[f_{X,Y}(X,Y)]^2$. We can decompose this squared joint density:
\[
[f_{X,Y}(X,Y)]^2 = [f_{Y\mid X}(Y\mid X) \cdot f_X(X)]^2 = [f_{Y\mid X}(Y\mid X)]^2 \cdot [f_X(X)]^2.
\]
This decomposition reveals a dual weighting mechanism. The \textbf{$[f_{Y\mid X}(Y\mid X)]^2$} term provides robustness to outliers in the response variable, operating identically to the conditional density kernel discussed in the main paper. Simultaneously, the \textbf{$[f_X(X)]^2$} term directly addresses high-leverage points. Observations $X$ that lie in low-density regions of the predictor space will have a small $f_X(X)$ value, and this effect is amplified by the squaring.

Therefore, the joint density kernel explicitly down-weights points that are unusual in either the response space (outliers) or the predictor space (high-leverage points). This provides a clear theoretical underpinning for its use in settings where both types of robust treatment are desired. A full investigation of its properties is left for future work.

\section{Asymptotic Bias Example with an Exponential Conditional Distribution}
\label{appendix:asym_bias_example}

This appendix illustrates how asymmetry in the conditional distribution \( f(Y \mid X) \) can introduce asymptotic bias in the proposed estimator. The focus is on an exponential distribution.

Suppose that for each fixed \(X\), the conditional distribution \( f(Y \mid X)\) follows an exponential distribution with rate parameter \(\lambda(X)\):
\[
 f(Y \mid X) \;=\; \lambda(X)\,\exp\!\bigl(-\lambda(X)\,Y\bigr),
 \quad Y \ge 0,
\]
so that the true regression function is
\[
 m(X) \;=\; \mathbb{E}[Y \,\mid\, X] \;=\; \frac{1}{\lambda(X)}.
\]
When this density is squared, we obtain
\[
 [f(Y \mid X)]^2 
 \;=\; [\lambda(X)]^2 
 \,\exp\!\bigl(-2\,\lambda(X)\,Y\bigr),
 \quad Y \ge 0,
\]
which is proportional to an exponential density with rate \(2\,\lambda(X)\). The mean of \(Y\) under this squared density is
\[
 \mu'(X) 
 \;=\; \frac{1}{2\,\lambda(X)}.
\]
As established in the main text, the oracle population target is \(\mu'(X)\) rather than \(m(X)\). Consequently, at each point \(x\), the asymptotic bias is
\[
 \text{Bias}(x)
 \;=\; \mu'(x) \;-\; m(x)
 \;=\; \frac{1}{2\,\lambda(x)} \;-\; \frac{1}{\lambda(x)}
 \;=\; -\frac{1}{2\,\lambda(x)}.
\]
This example illustrates how the asymmetry of an exponential distribution can steer the estimator toward \(1/(2\,\lambda(X))\) rather than the true mean \(1/\lambda(X)\). Although such a shift introduces asymptotic bias, the robust weighting can still be advantageous in practical situations where outliers or heavy-tailed noise are significant concerns.

\end{document}